\RequirePackage{etoolbox}
\csdef{input@path}{%
	{img/}% eps files
}%

\documentclass[preprint]{imsart}

\usepackage{times}
\usepackage{booktabs} % Horizontal rules in tables
\usepackage{amsthm}
\usepackage{amsmath}
\usepackage{physics}
\usepackage{amssymb}
\usepackage{amstext} % for \text macro
\usepackage{array}   % for \newcolumntype macro
\usepackage{enumitem} % Customized lists
\setlist[itemize]{noitemsep} % Make itemize lists more compact
\allowdisplaybreaks
\usepackage{multirow}
\usepackage[flushleft]{threeparttable}
\usepackage{verbatim}
\usepackage{hyperref}
\usepackage{bbm}
\usepackage{longtable}
\usepackage[figuresright]{rotating}
\usepackage{caption}
\usepackage{graphicx}
\usepackage{subcaption}
\usepackage{indentfirst}

\usepackage{natbib}

\newcommand{\R}{\mathbb{R}}
\newtheorem{theorem}{Theorem}
\newtheorem{remark}{Remark}

\usepackage{calrsfs}
\DeclareMathAlphabet{\pazocal}{OMS}{zplm}{m}{n}
\newcommand{\Na}{\pazocal{N}}
\newcommand{\Wa}{\pazocal{W}}

\newcommand{\Xa}{\pazocal{X}}

\startlocaldefs
\numberwithin{equation}{section}
\theoremstyle{plain}

\endlocaldefs

\begin{document}

\begin{frontmatter}
\title{Bayesian Classification of Multiclass Functional Data}
\runtitle{Functional Data Classification}

\begin{aug}
	\author{\fnms{Xiuqi} \snm{Li}\ead[label=e1]{xli35@ncsu.edu}}
	\address{Operations Research Graduate Program \\
		North Carolina State University\\
		Raleigh, NC 27695\\
		U.S.A.\\
		\printead{e1}}
	\and
	\author{\fnms{Subhashis} \snm{Ghosal}\ead[label=e2]{sghosal@stat.ncsu.edu}}
	\address{Department of Statistics\\
		North Carolina State University\\
		Raleigh, NC 27695\\
		U.S.A.\\
		\printead{e2}}
	
	\runauthor{Li and Ghosal}
	
\end{aug}

\begin{abstract}
	We propose a Bayesian approach to estimating parameters in multiclass functional models. Unordered multinomial probit, ordered multinomial probit and multinomial logistic models are considered. We use finite random series priors based on a suitable basis such as B-splines in these three multinomial models, and classify the functional data using the Bayes rule. We average over models based on the marginal likelihood estimated from Markov Chain Monte Carlo (MCMC) output. Posterior contraction rates for the three multinomial models are computed. We also consider Bayesian linear and quadratic discriminant analyses on the multivariate data obtained by applying a functional principal component technique on the original functional data. A simulation study is conducted to compare these methods on different types of data. We also apply these methods to a phoneme dataset.  
\end{abstract}


%\tableofcontents
\end{frontmatter}

\section{Introduction}

\par Functional data analysis (FDA) deals with the analysis of data occurring in the form of functions. \cite{WangJane-Ling} gave an overview of FDA including functional principal component analysis, functional linear regression, clustering and classification of functional data. FDA is increasingly drawing attention in many areas, such as biomedicine, environmental studies, and economics \citep{Ullah}. \cite{Mallor} proposed a model based on functional principal component analysis to predict household electricity consumption. \cite{Wagner-Muns} proposed a method that uses functional principal components analysis to forecast traffic volume. Classification of functional data, especially when the data units can come from more than two categories, is a fundamental problem of interest. Generalized linear models are often used to classify the functional data \citep{Muller,James2002}. The linear discriminant analysis is also used for functional data classification \citep{James2001}. \cite{Preda} proposed the partial least squares regression on functional data for linear discriminant analysis. \cite{Rossi} adapted support vector machines to functional data classification. \cite{Li2008} proposed a functional segment discriminant analysis (FSDA), which combines the classical linear discriminant analysis and support vector machines. Wavelets approaches are also applied to classify and cluster functional data \citep{Ray,Antoniadis,Chang,Suarez}.
There are also nonparametric approaches for functional data classification \citep{Biau,Ferraty}. However, there are only a few approaches proposed in the context of Bayesian classification for functional data. \cite{Wang2007} developed a Bayesian hierarchical model which combines the adaptive wavelet-based function estimation and the logistic classification. \cite{Zhu2010} proposed a Bayesian hierarchical model that takes into account random batch effects and selects effective functions among multiple functional predictors. \cite{Stingo} proposed a Bayesian conjugate normal discriminant model on the wavelet transform of the functional data. \cite{Zhu} introduced two Bayesian approaches: the Gaussian, wavelet-based functional mixed model and the robust, wavelet-based functional mixed model. 
\par In this paper, we consider a response $Y$ taking values $k=1,\dots,K$, with functional covariate $\{X(t), t\in [0,1]\}$. The main problem is to estimate the probability $\text{P}(Y=k|X)$, which can be conveniently modeled by a function of $\int \beta(t)X(t)dt$
\begin{align}
	\text{P}(Y=k|X)=H_k\left(\int \beta(t)X(t)dt\right),
\end{align}
where $H_k$ is a cumulative distribution function, and $\beta(\cdot)$ is an unknown (possibly vector of) coefficient function(s). Unordered multinomial probit, ordered multinomial probit and multinomial logistic models are considered in this paper which correspond to different choices of $H_k, k=1,\dots,K$. For an ordered multinomial probit model, there are additional order restrictions. Finite random series priors \citep{Shen} are applied to the three multinomial models. We compare these methods with Bayesian linear and quadratic discriminant analyses applied on the data reduced to multivariate form by a functional principal component technique. Following a Bayesian approach, the posterior distribution of the parameters are obtained using the training data, and then the classification rules are applied to the test data using the posterior probability of class membership.
\par  The primary goal of a basis expansion method is to reduce a more complex problem to a simpler problem which has either a known solution or is likely to be easier to solve. A prior on function through finite random series is a standard tool in nonparametric Bayesian inference, but in the context of functional data, the technique has not been utilized to its fullest potential, especially regarding the study of theoretical property of Bayesian methods. Only one paper \citep{Shen} has one example of functional linear regression treated using finite random series priors. We take that idea but develop it in the context of functional data classification. Characterizing contraction rate is a major goal of this paper. For this, we need to estimate the complexity of the model and the prior concentration. Even though, the model reduces to the finite dimensional setting from the computational point of view, the effect of the residual bias in the approximation of function must be properly addressed. Hence the treatment substantially different from that of a parametric problem. In particular, the dimension of the basis must be adapted with the smoothness and the sample size by using a prior on it.
\par The paper is organized as follow. In Section 2, the 
three functional multinomial models are described. Section 3 gives the description of applying the finite random series prior to these models. The marginal likelihood estimation is described in Section 4. In Section 5, the posterior contraction rates of the three functional multinomial models with finite random series prior are computed. Section 6 describes the Bayesian discriminant analysis of functional data, which is used to compare with the proposed models. In Section 7, a simulation study is conducted on various types of data. In Section 8, the three multinomial models and Bayesian discriminant analysis are tested on phoneme dataset.  

\section{Model}

\subsection{Ordered Multinomial Probit Model}
\par Let $X_{i}(t)$, $i=1,\dots,n$, $t\in[0,1]$, be the orbserved functional data associated with a categorical variable $Y_{i}$ taking possible values $1,\dots,K$. We assume that $(X_i,Y_i)$, $i=1,\dots,n$, are independent and identically distributed (i.i.d) observations.
\par Following \cite{chib1993}, we consider the model described implicitly as follows: there exists a latent variable $W_{i}$ distributed as $\text{N}(\int \beta(t)X_{i}(t)dt,1)$, for $i=1,\dots,n$, and that $Y_{i}=k$ if $\gamma_{k-1}<W_{i}\leq \gamma_{k}$, where $k=1,\dots,K$. The latent variables $W_i$, $i=1,\dots,n$, are independent. The coefficient function $\beta(\cdot)$ is unkown. The cut-points $\gamma_k$ are also unknown except that $\gamma_{0}=-\infty$ and $\gamma_{K}=\infty$. To ensure identifiability, we set $\gamma_{1}=0$. Under the assumed model, the probability of choosing a category $k$ is given by
\begin{align}
	\text{P}(Y_{i}=k|X_{i})=\Phi\left(\gamma_{k}-\int \beta(t)X_{i}(t)dt\right)-\Phi\left(\gamma_{k-1}-\int \beta(t)X_{i}(t)dt\right),
\end{align}
where $\Phi$ stands for the distribution function of the standard normal distribution.

\subsection{Unordered Multinomial Probit Model}
\par Let $X_i(t)$, $i=1,\dots,n$, be the same as in the Section 2.1, and also same for Section 2.3. 
\par The unordered multinomial probit model can be described by the following data augmentation method. As in \cite{chib1993}, let $W_{i}'=(W_{i1}',\dots,W_{iK}')^{T}$, $i=1,\dots,n$, be latent variable, such that $W_{il}'$ follows a linear model
\begin{align}
	W_{il}'=\int\beta_{l}'(t)X_{i}(t)dt+\varepsilon_{il}',
\end{align}
where $\varepsilon_{il}'\sim \text{N}(0,1)$, $i=1,\dots,n$, $l=1,\dots,K$, are i.i.d. standard normal random variables.
Consider the latent variables $W_{i}=(W_{i1},\dots,W_{iK-1})^{T}$,  $W_{il}=W_{il}'-W_{iK}'$, 
\begin{align}
	W_{il}=\int\beta_{l}'(t)X_{i}(t)dt-\int \beta_{K}'(t)X_{i}(t)dt+\varepsilon_{il},
\end{align}
where $\varepsilon_{il}=\varepsilon_{il}'-\varepsilon_{iK}'$, and $l=1,\dots,K-1$. Let $\varepsilon_{i}=(\varepsilon_{i1},\dots,\varepsilon_{iK-1})^{T}$. Then $\varepsilon_{i}$ follows $\text{N}(0,\Sigma)$, where $\Sigma$ is a $(K-1)\times(K-1)$ matrix with 2 at diagonal entries and 1 at all off-diagonal entries.
\par The probability of choosing the $k$th ($k=1,\dots,K-1$) alternative is given by
\begin{align}
	\text{P}(Y_{i}=k|X_{i})=\text{P}(W_{ik}>W_{il}, \:\text{for all}\:l\neq k, \text{and} \: W_{ik}>0),
\end{align}
and the probability of choosing alternative $K$ is given by
\begin{align}
	\text{P}(Y_{i}=K|X_{i})=\text{P}(W_{il}<0 \:\text{for all} \: l=1,\dots,K-1).
\end{align}

\subsection{Multinomial Logistic Model}
\par In this model, the probability of choosing category $k$ is given by
\begin{align}
	\text{P}(Y_{i}=k|X_{i})=\frac{\exp[\int \beta_{k}(t)X_{i}(t)dt]}{\sum_{l=1}^{K}\exp[\int \beta_{l}(t)X_{i}(t)dt]}.
\end{align}
\par To ensure model identification, set $\beta_{K}(t)=0$. Then the probability of choosing categoty $k$ ($k=1,\dots,K-1$) is given by
\begin{align}
	\text{P}(Y_{i}=k|X_{i})=\frac{\exp[\int \beta_{k}(t)X_{i}(t)dt]}{1+\sum_{l=1}^{K-1}\exp[\int \beta_{l}(t)X_{i}(t)dt]},
\end{align}
and the probability of choosing category $K$ is given by
\begin{align}
	\text{P}(Y_{i}=K|X_{i})=\frac{1}{1+\sum_{l=1}^{K-1}\exp[\int \beta_{l}(t)X_{i}(t)dt]}.
\end{align}

\section{Finite Random Series Prior}
\par The functional coefficient $\beta(t)$ (or $\beta_1(t),\dots,\beta_K(t)$ for unordered multinomial probit and multinomial logistic model ) is given a prior which is a finite linear combination of a certain chosen basis functions: $\beta(t)=\sum_{j=1}^{J} \theta_{j}\psi_{j}(t)$, where $\{\psi_{1}(t),\dots,\psi_{J}(t)\}$ is a basis, for example, formed by B-splines, Fourier functions, or wavelets. A prior is put on the unknown coefficients $(\theta_1,\dots,\theta_J)$. The number of basis function $J$ is also unknown and should be given a hyperprior. Instead of sampling across the different dimensions using reversible jump MCMC \citep{Green} which has computational difficulty for complicated models, we can implement MCMC for a given $J$ value, and repeat it for relevant $J$ values. Thus, we can compute the marginal likelihood $m(Y|J)$ for potentially interesting values of $J$, and obtain the posterior probability of $J$, which are discussed in Section 4.
\par The advantage of a using finite random series prior is that the inner product between the functional coefficient and the functional data $\int \beta(t)X_{i}(t)dt$ is reduced to a simple linear combination
\begin{align}
	\int \beta(t)X_{i}(t)dt=\int \sum_{j=1}^{J} \theta_{j}\psi_{j}(t)X_{i}(t)dt=\sum_{j=1}^{J} \theta_{j}Z_{ij},
\end{align}
where $Z_{ij}=\int \psi_{j}(t)X_{i}(t)dt$ is known, and can be computed by Simpson's rule.

\subsection{Ordered Multinomial Probit Model}
\par Using a finite random series $\beta(t)=\sum_{j=1}^{J}\theta_{j}\psi_{j}(t)$, the model in (2.1) can be rewritten as
\begin{align}
	\text{P}(Y_{i}=k|X_{i})=\Phi\left(\gamma_{k}-\sum_{j=1}^{J} \theta_{j}Z_{ij}\right)-\Phi\left(\gamma_{k-1}-\sum_{j=1}^{J} \theta_{j}Z_{ij}\right),
\end{align}
where $Z_{ij}=\int \psi_{j}(t)X_{i}(t)dt$.
\par Define $\theta=(\theta_{1},\dots,\theta_{J})^{T}$, and $Z_{i}=(Z_{i1},\dots,Z_{iJ})^{T}$. Then (3.2) can be written compactly as
\begin{align}
	\text{P}(Y_{i}=k|X_{i})=\Phi(\gamma_{k}-Z_{i}^{T}\theta)-\Phi(\gamma_{k-1}-Z_{i}^{T}\theta).
\end{align}
Clearly the unobserved latent variable $W_{i}$ follows $\text{N}_J(Z_{i}^{T}\theta,1)$, where $\text{N}_J$ stands for the $J$-variate normal distribution. Assign a conjugate prior $\theta\sim\text{N}_J(\theta_{0},B_{0})$, where $\theta_0$ is $J\times 1$ mean vector, and $B_{0}$ is a $J\times J$ covariance matrix. Then the posterior distribution of $\theta$ is given by
\begin{align}
	\theta|Y,W\sim \text{N}_J(\theta_{n},B_{n}),\: B_{n}=(B_{0}^{-1}+Z^{T}Z)^{-1},\: \theta_{n}=B_{n}(B_{0}^{-1}\theta_{0}+Z^{T}W),
\end{align}
where $Z=(Z_{1}^{T},\dots,Z_{n}^{T})^{T}$, and $W=(W_{1},\dots,W_{n})^{T}$.
\par We follow the scheme introduced by \cite{chib1993}. The posterior distribution of $W_{i}$ is given by
\begin{align}
	W_{i}|(\theta,\gamma,Y_{i}=k)\sim \text{TN}(Z_{i}^{T}\theta,1,\gamma_{k-1},\gamma_{k}),
\end{align}
where $\text{TN}(Z_{i}^{T}\theta,1,\gamma_{k-1},\gamma_{k})$ is the truncation of the (univariate) normal distribution with mean $Z_{i}^{T}\theta$, and variance 1 to the interval $(\gamma_{k-1}, \gamma_{k})$.
\par \cite{chib1993} assigned a diffuse prior on the cut-points. However, model averaging needs a proper prior. A normal prior is not appropriate due to the order restriction on $\gamma_1,\dots,\gamma_K$. \cite{chib1997} proposed a transformation of $\gamma=(\gamma_1,\dots,\gamma_K)$ which avoids the order restriction.
\begin{align}
	\alpha_1=\log\gamma_2, \: \alpha_j=\log(\gamma_{j+1}-\gamma_j), \: 2\leq j \leq K-2.
\end{align}
\par Note that $\gamma_1=0$ and by the inverse map
\begin{align}
	\gamma_j=\sum_{l=1}^{j-1}e^{\alpha_l},\: 2\leq j \leq K-1.
\end{align}
Then $\gamma$ can be reparameterized by $\alpha=(\alpha_1,\dots,\alpha_{K-2})$. Assign a multivariate normal prior with mean $\alpha_0$, and covariance $A_0$ on $\alpha$. To sample $\gamma$, apply the following steps of Metropolis-Hastings algorithm.
\begin{enumerate}
	\item Sample $\alpha'$ from a proposal distribution $q(\alpha',\alpha| Y,\theta,W)$. Here we allow the proposal density to depend on the data and the two remaining blocks for the convenience of computing the marginal likelihood in the future.
	\item Move to $\alpha'$ from the current $\alpha$ with probability
	\begin{align}
		\rho(\alpha, \alpha'| Y,\theta,W)=\text{min}\Bigl\{\frac{f(Y|\alpha',\theta,W)\pi(\alpha',\theta)}{f(Y|\alpha,\theta,W)\pi(\alpha,\theta)}\frac{q(\alpha',\alpha| Y,\theta,W)}{q(\alpha,\alpha'| Y,\theta,W)}, 1 \Bigr\}.
	\end{align}
	\item Compute $\gamma$ by the inverse map (3.7).
\end{enumerate}
\par To implement the MCMC sampling, first draw $\gamma$ by the above steps. Then sample from the posterior distributions (3.5) and (3.4).
\par The values of $\gamma$ sampled from the Metropolis-Hastings algorithm converges quickly. We demonstrate it on the real data in Section 8 by plotting the sampling values of $\gamma$.

\subsection{Unordered Multinomial Probit Model}
\par Let $\beta_{l}'(t)=\sum_{j=1}^{J} \theta_{lj}'\psi_{j}(t)$, where $l=1,\dots,K$. Then (2.3) can be rewritten as
\begin{align}
	W_{il}=\sum_{j=1}^{J}\theta_{lj}'Z_{ij}-\sum_{j=1}^{J}\theta_{Kj}'Z_{ij}+\epsilon_{il}=\sum_{l=1}^{J}(\theta_{jl}'-\theta_{jK}')Z_{ij}+\varepsilon_{il},
\end{align}
where $Z_{ij}=\int \psi_{j}(t)X_{i}(t)dt$.
\par Let $\theta_{lj}=\theta_{lj}'-\theta_{Kj}'$, where $j=1,\dots,J$. Define $\theta_{l}=(\theta_{l1},\dots,\theta_{lJ})^{T}$, and $Z_{i}=(Z_{i1},\dots,Z_{iJ})^{T}$. Then (3.9) is given by
\begin{align}
	W_{il}=Z_{i}^{T}\theta_{l}+\varepsilon_{il},
\end{align}
where $i=1,\dots,n$, $l=1,\dots,K-1$.
\par Define a $J\times(K-1)$ matrix $\Theta=(\theta_{1},\dots,\theta_{K-1})$. Then we have $W_{i}=Z_{i}^{T}\Theta+\varepsilon_{i}$, where $W_{i}=(W_{i1},\dots,W_{iK-1})^{T}$, $\varepsilon_{i}=(\varepsilon_{i1},\dots,\varepsilon_{iK-1})^{T}$, and $\varepsilon_{i}\sim \text{N}(0,\Sigma)$.
\par In the model described in Section 2, $\Sigma$ is known with 2 on diagonal entries and 1 on all off-diagonal entries. The only  parameter needs to be estimated is $\Theta$. In order to draw the matrix $\Theta$ using the Gibbs sampling, we can stack the data in a matrix form which is given by
\begin{align}
	W=Z\Theta+\varepsilon,
\end{align}
where $W=(W_{1}^{T},\dots,W_{n}^{T})^{T}$ is an $n\times(K-1)$ matrix, $Z=(Z_{1}^{T},\dots,Z_{n}^{T})^{T}$ is an $n\times J$ matrix, and $\varepsilon=(\varepsilon_{1}^{T},\dots,\varepsilon_{n}^{T})^{T}$ is an $n\times (K-1)$ matrix.
\par This results in a matrix normal distribution. The density function of matrix normal distribution $\text{MN}_{n\times p}(M,U,V)$ is
\begin{align}
	(2\pi)^{-np/2}|V|^{-n/2}|U|^{-p/2}\exp\left(-\frac{1}{2}\text{tr}[V^{-1}(X-M)^{T}U^{-1}(X-M)]\right),
\end{align}
where $M$ is an $n\times p$ mean matrix, $U$ is an $n\times n$ row variance matrix, $V$ is a $p\times p$ column variance matrix, $\text{tr}$ stands for the trace of a matrix, and $|U|$ and $|V|$ denote the determinants of $U$ and $V$ respectively.
\par Thus $W|\Theta \sim \text{MN}_{n\times (K-1)}(Z\Theta,I_{n},\Sigma)$. Here the row variance-covariance matrix $I_{n}$ is an identity matrix of rank $n$, since $W_{1},\dots,W_{n}$ are independent. 
We consider the matrix normal prior $\Theta\sim \text{MN}_{J\times(K-1)}(U_{0},V_{0},\Sigma)$. By a standard conjugacy calculation, the posterior is given by
\begin{align}
	\label{eqn:eqlabel}
	\begin{split}
		&\Theta|Y,W\sim \text{MN}_{J\times(K-1)}(U_{n},V_{n},\Sigma),\\ &V_{n}=(Z^{T}Z+V_{0}^{-1})^{-1},\: U_{n}=V_{n}(Z^{T}W+V_{0}^{-1}U_{0}).
	\end{split}
\end{align}
\par To draw a sample of $W$, we use the method introduced by \cite{McCulloch}. Let $W_{i,-l}$ denote  $(W_{i1},\dots,W_{i,l-1},W_{i,l+1},\dots,W_{iK-1})^{T}$, $Z_{i,\cdot}$ denote the $i$th row of $Z$, the vector $\Theta_{\cdot,l}$ denote the $l$th column of $\Theta$, the matrix $\Theta_{\cdot,-l}$ denote $\Theta$ without the $l$th column, the scalar $\Sigma_{l,l}$ denote the $(l,l)$th entry of $\Sigma$, $\Sigma_{-l,-l}$ denote $\Sigma$ without the $l$th row and the $l$th column, $\Sigma_{-l,l}$ denote the $l$th column of $\Sigma$ without the $l$th entry, and $\Sigma_{l,-l}$ denote the $l$th row of $\Sigma$ without the $l$th entry. We draw $W_{il}$ from the conditional truncation of the normal distribution with the mean $m_{il}$ and variance $\tau_{il}^2$ to the interval $(a,b)$ described below:
\begin{align}
	\label{eqn:eqlabel}
	\begin{split}
		&W_{il}|(W_{i,-l},\Theta,Y_{i})\sim \text{TN}(m_{il},\tau_{il}^2,a,b),\\
		&m_{il}=Z_{i,\cdot}\Theta_{\cdot, l}+\Sigma_{-l,l}^{T}\Sigma_{-l,-l}^{-1}(W_{i,-l}-Z_{i, \cdot}\Theta_{\cdot,-l}),\\
		&\tau_{il}^{2}=\Sigma_{l,l}-\Sigma_{l,-l}\Sigma_{-l,-l}^{-1}\Sigma_{-l,l},\\
		&(a,b)=
		\begin{cases}
			(\text{max} \{W_{i,-l},0\},\infty), & \text{if}\ Y_{i}=l, \ l=1,\dots,K-1, \\
			(-\infty,\text{max}\{W_{i,-l}\}), & \text{if}\ Y_{i}\neq l, \ l=1,\dots,K-1,\\
			(-\infty,0), & \text{if}\ Y_{i}=K,
		\end{cases}\\
		&i=1,\dots,n, \: l=1,\dots,K-1. 
	\end{split}
\end{align}
\par To implement the Gibbs sampling, we draw samples from (3.13) and (3.14).

\subsection{Multinomial Logistic Model}
Let $\beta_{k}(t)=\sum_{j=1}^{J}\theta_{kj}\psi_{j}(t)$. Then (2.7) and (2.8) can be rewritten as
\begin{align}
	\text{P}(Y_{i}=k|X_{i})=\frac{\exp[\sum_{j=1}^{J} \theta_{kj}Z_{ij}]}{1+\sum_{l=1}^{K-1}\exp[\sum_{j=1}^{J} \theta_{lj}Z_{ij}]},\: k=1,\dots,K-1,
\end{align}
\begin{align}
	\text{P}(Y_{i}=K|X_{i})=\frac{1}{1+\sum_{l=1}^{K-1}\exp[\sum_{j=1}^{J} \theta_{lj}Z_{ij}]},
\end{align}
where $Z_{ij}=\int \psi_{j}(t)X_{i}(t)dt$.
\par Define $\theta_{k}=(\theta_{k1},\dots,\theta_{kJ})^{T}$,$\:k=1,\dots,K-1$, and $Z_{i}=(Z_{i1},\dots,Z_{iJ})^{T}$. Then (3.15) and (3.16) are given by
\begin{align}
	\text{P}(Y_{i}=k|X_{i})=\frac{\exp[Z_{i}^{T}\theta_{k}]}{1+\sum_{l=1}^{K-1}\exp[Z_{i}^{T}\theta_{l}]}, \: k=1,\dots,K-1,
\end{align}
\begin{align}
	\text{P}(Y_{i}=K|X_{i})=\frac{1}{1+\sum_{l=1}^{K-1}\exp[Z_{i}^{T}\theta_{l}]}.
\end{align}
\par For each $\theta_{k}$, $k=1,\dots,K-1$, we assign a multivariate normal prior $\text{N}_J(\mu_{k},\Sigma_{k})$, and apply Metropolis-Hastings algorithm to sample $\theta_{k}$.
\begin{enumerate}
	\item Sample $\theta_k'$ from the proposal distribution $q(\theta_k',\theta_k| Y,\theta_{-k})$.
	\item Move to $\theta_k'$ from the current $\theta_k$ with probability
	\begin{align}
		\rho(\theta_k, \theta_k'| Y,\theta_{-k})=\text{min}\Bigl\{\frac{f(Y|\theta_k',\theta_{-k})\pi(\theta_k',\theta_{-k})}{f(Y|\theta_k,\theta_{-k})\pi(\theta_k,\theta_{-k})}\frac{q(\theta_k',\theta_k| Y,\theta_{-k})}{q(\theta_k,\theta_k'| Y,\theta_{-k})}, 1 \Bigr\},
	\end{align}
\end{enumerate}
where $\theta_{-k}$ denotes all the blocks except the $k$th one.

\section{Marginal Likelihood and Model Averaging}
\par In Section 3, we described the MCMC sampling technique for a given $J$ value, which we need to repeat it for all possible $J$ values. In the actual computation, however, it is impossible to consider all values of $J$. With a given prior on $J$, for example, geometric or Poisson distribution, the posterior probabilities for very small or very large values of $J$ decay to zero very quickly. Thus, we do not need to consider these $J$ values. Let $J_1,\dots,J_S$ denote the values of $J$ we need to consider. If we can get the marginal likelihood $m(Y|J_s)$, then we can compute the posterior probability of $J_s$ using Bayes's rule
\begin{align}
	\text{P}(J_s|Y)=\frac{m(Y|J_s)p(J_s)}{\sum_{l=1}^{S}m(Y|J_l)p(J_l)},
\end{align}
where $p(J_s)$, $s=1,\dots,S$, is the prior probability for $J=J_s$.
\par For each given $J_s$, we have a misclassification rate $r_s$, which is defined as the ratio of the number of falsely classified data to the total number of data. Then we can obtain the average misclassification rate $\bar{r}$ for each multinomial model:
\begin{align}
	\bar{r}=\sum_{s=1}^{S}\text{P}(J_s|Y)\cdot r_s.
\end{align}
We call it the model averaging method.
\par The marginal likelihood can be written as the normalizing constant of the posterior density
\begin{align}
	m(Y|J_s)=\frac{f(Y|J_s,B)\pi(B|J_s)}{\pi(B|Y,J_s)},
\end{align}
where $B$ is a convenient value of the parameter in the context of the support of the posterior distribution such as the posterior mean, because (4.3) holds for any $B$. The numerator is the product of the likelihood and the prior. The denominator is the posterior density of $B$. For a given $B^*$, the posterior density $\pi(B^*|Y, J_m)$ can be estimated from the Gibbs output \citep{Chib1995} and the Metropolis-Hasting output \citep{Chib2001}. Then the estimated marginal likelihood in the logarithm scale is
\begin{align}
	\log \hat{m}(Y|J_s)=\log f(Y|J_s,B^*)+\log \pi(B^*|J_s)-\log \hat{\pi}(B^*|Y,J_s).
\end{align}
The following sections give the details for $\pi(B^*|Y,J_s)$ estimation.
\subsection{Ordered Multinomial Probit Model}
\par There are two parameter blocks in this model, $\theta$ and $\alpha$, where $\alpha$ is the transformation of $\gamma$ as in (3.6). Given $\theta^*=G^{-1}\sum_{g=1}^{G}\theta^{(g)}$, and $\alpha^*=G^{-1}\sum_{g=1}^{G}\alpha^{(g)}$, where $\{\theta^{(g)},\alpha^{(g)}\}_{g=1}^{G}$ are from the MCMC output, the joint posterior density can be written as
\begin{align}
	\pi(\theta^*,\alpha^*|Y,J_s)=\pi(\alpha^*|Y,J_s)\pi(\theta^*|Y,J_s,\alpha^*),
\end{align}
where 
\begin{align}
	\pi(\theta^*|Y,J_s,\alpha^*)=\int \pi(\theta^*|Y,J_s,\alpha^*,W)\pi(W|Y,J_s,\alpha^*)dW.
\end{align}
\par The Monte Carlo estimate of $\pi(\theta^*|Y,J_s,\alpha^*)$ is
\begin{align}
	\hat{\pi}(\theta^*|Y,J_s,\alpha^*)=M^{-1}\sum_{m=1}^{M}\pi(\theta^*|Y,J_s,\alpha^*,W^{(m)}),
\end{align}
where $\{W^{(m)}\}_{m=1}^{M}$ are sampled from distribution $[W|Y,J_s,\alpha^*]$. The draws of $W$ from the Gibbs sampler are from the distribution $[W|Y,J_s]$, so $\pi(\theta^*|Y,J_s,\alpha^*,W)$ cannot be averaged directly by the Gibbs sampling output. Addtional sampling for $W$ is needed. We sample $\{\theta^{(m)}\}$ from the density  $\pi(\theta|Y,J_s,\alpha^*,W)$, and given that $\theta^{(m)}$, we sample $\{W^{(m)}\}$ from $\pi(W|Y,J_s,\theta,\alpha^*)$. 
\par The explicit distribution of $\alpha^*$ given $(Y,J_s)$ is unknown, and hence the draws of $\alpha$ are obtained from a Metropolis-Hastings sampling. By the local reversibility condition (see \cite{Chib2001} for details), the posterior density of $\alpha$ can be written as
\begin{align}
	\pi(\alpha^*|Y,J_s)=\frac{\text{E}_1\{\rho(\alpha,\alpha^*|Y,J_s,\theta,W)q(\alpha,\alpha^*|Y,J_s,\theta,W)\}}{\text{E}_2\{\rho(\alpha^*,\alpha|Y,J_s,\theta,W)\}},
\end{align}
where $\rho(\alpha,\alpha^*|Y,J_s,\theta,W)$ is defined in (3.8), $q(\alpha,\alpha^*|Y,J_s,\theta,W)$ is the proposal density, the expectation $\text{E}_1$ is with respect to the distribution $\pi(\theta,\alpha,W|Y,J_s)$, and $\text{E}_2$ is with respect to the distribution $\pi(\theta,W|Y,J_s,\alpha^*)\times q(\alpha^*,\alpha|Y,J_s,\theta,W)$.
\par Then an estimate of $\pi(\alpha^*|Y,J_s)$ is given by
\begin{align}
	\frac{G^{-1}\sum_{g=1}^{G}\rho(\alpha^{(g)},\alpha^*|Y,J_s,\theta^{(g)},W^{(g)})q(\alpha^{(g)},\alpha^*|Y,J_s,\theta^{(g)},W^{(g)})}{M^{-1}\sum_{m=1}^{M}\rho(\alpha^*,\alpha^{(m)}|Y,J_s,\theta^{(m)},W^{(m)})},
\end{align}
where $\{\theta^{(g)},\alpha^{(g)},W^{(g)}\}_{g=1}^{G}$ are obtained from the MCMC output. $\{\theta^{(m)},W^{(m)}\}$ are obtained from $\pi(\theta|Y,J_s,\alpha^*,W)$ and $\pi(W|Y,J_s,\theta,\alpha^*)$, and then given $\{\theta^{(m)},W^{(m)}\}$, sample $\alpha^{(m)}$ from $q(\alpha^*, \alpha|Y,J_s,\theta^{(m)},W^{(m)})$.

\subsection{Unordered Multinomial Probit Model}
\par The only unknown parameter is $\Theta$. For $\Theta^*=G^{-1}\sum_{g=1}^{G}\Theta^{(g)}$, where $\{\Theta^{(g)}\}$ are from the Gibbs sampling output, the posterior density of $\Theta$ at $\Theta^*$ can be written as
\begin{align}
	\pi(\Theta^*|Y,J_s)=\int \pi(\Theta^*|Y,J_s,W)\pi(W|Y,J_s)dW.
\end{align}
Then the Monte Carlo estimate of $\pi(\Theta^*|Y,J_s)$ is
\begin{align}
	\hat{\pi}(\Theta^*|Y,J_s)=\sum_{g=1}^{G} \pi(\Theta^*|Y,J_s,W^{(g)}),
\end{align}
where $\{W^{(g)}\}_{g=1}^{G}$ are from the Gibbs sampling output.
\par For the unordered multinomial probit model, we also need to estimate the likelihood at some convenient values in the support of the posterior distribution. From Section 3.2, $\Theta=(\theta_1,\dots,\theta_{K-1})$, where $\theta_l=\theta_l'-\theta_K'$, $l=1,\dots,K-1$. Then (2.5) can be rewritten as
\begin{align}
	\label{eqn:eqlabel}
	\begin{split}
		&\text{P}(Y=K)\\
		&=\frac{1}{(2\pi)^{(K-1)/2}|\Sigma|^{1/2}}\int_{-\infty}^{-Z^T\Theta_{\cdot,1}}\dots\int_{-\infty}^{-Z^T\Theta_{\cdot,K-1}}\exp\big(-\frac{1}{2}U^T\Sigma^{-1}U\big)dU,
	\end{split}
\end{align}
where $\Theta_{\cdot,l}$ denotes the $l$th column of $\Theta$.
\par For $l\neq K$, let $\Theta^l=(\theta_1-\theta_l,\dots,\theta_{l-1}-\theta_l,\theta_{l+1}-\theta_l,\dots,\theta_{K-1}-\theta_l,-\theta_l)$, then
\begin{align}
	\label{eqn:eqlabel}
	\begin{split}
		&\text{P}(Y=l)\\
		&=\frac{1}{(2\pi)^{(K-1)/2}|\Sigma|^{1/2}}\int_{-\infty}^{-Z^T\Theta_{\cdot,1}^l}\dots\int_{-\infty}^{-Z^T\Theta_{\cdot,K-1}^l}\exp\big(-\frac{1}{2}U^T\Sigma^{-1}U\big)dU.
	\end{split}
\end{align}
\par Due to the exchangeable correlation structure of $\Sigma$, (4.13)  can be reduced to a one dimensional integral \citep{Dunnett} given by
\begin{align}
	\label{eqn:eqlabel}
	\begin{split}
		&\text{P}(Y=l)\\
		&=\frac{1}{\sqrt{\pi}}\int_{0}^{\infty}\bigl\{\prod_{k=1}^{K-1}\Phi(-u\sqrt{2}-Z^T\Theta_{\cdot,k}^{l})+\prod_{k=1}^{K-1}\Phi(u\sqrt{2}-Z^T\Theta_{\cdot,k}^{l})\bigr\}e^{-u^2}du.
	\end{split}
\end{align}
\par The expression in (4.12) can also be reduced to the same form as in (4.14). Then (4.14) can be approximated by Gaussian quadrature as follow
\begin{align}
	\text{P}(Y=l)\approx\frac{1}{2}w_q\bigl\{\prod_{k=1}^{K-1}\Phi(-\sqrt{2x_q}-Z^T\Theta_{\cdot,k}^{l})+\prod_{k=1}^{K-1}\Phi(\sqrt{2x_q}-Z^T\Theta_{\cdot,k}^{l})\bigr\},
\end{align}
where $w_q$ and $x_q$ are the weights and roots of the Laguerre polynomial of order $Q$.
\par Thus, the likelihood of this unordered multinomial probit model can be approxiamted using (4.15).

\subsection{Multinomial Logistic Model}
\par There are $K-1$ unknown parameters: $\theta_1,\dots,\theta_{K-1}$. Given $\theta_k^*=G^{-1}\sum_{g=1}^{G}\theta_k^{(g)}$, $k=1,\dots,K-1$, where $\{\theta_k^{(g)}\}_{g=1}^{G}$ are from the Metropolis-Hastings sampling output, the joint posterior density can be written as
\begin{align}
	\pi(\theta_1^*,\dots,\theta_{K-1}^*|Y,J_s)=\prod_{i=1}^{K-1}\pi(\theta_i|Y,J_s,\theta_1^*,\dots,\theta_{i-1}^*).
\end{align}
\par By the local reversibility, each full conditional density can be written as
\begin{align}
	\label{eqn:eqlabel}
	\begin{split}
		&\pi(\theta_i|Y,J_s,\theta_1^*,\dots,\theta_{i-1}^*)\\
		&=\frac{\text{E}_1\{\rho(\theta_i,\theta_i^*|Y,J_s,\Psi_{i-1}^*,\Psi^{i+1})q(\theta_i,\theta_i^*|Y,J_s,\Psi_{i-1}^*,\Psi^{i+1})\}}{\text{E}_2\{\rho(\theta_i^*,\theta_i|Y,J_s,\Psi_{i-1}^*,\Psi^{i+1})\}},
	\end{split}
\end{align}
where $\Psi_{i-1}=(\theta_1,\dots,\theta_{i-1})$, $\Psi^{i+1}=(\theta_{i+1},\dots,\theta_{K-1})$, $\rho(\theta_i,\theta_i^*|Y,J_s,\Psi_{i-1}^*,\Psi^{i+1})$ is defined in (3.19), $q(\theta_i,\theta_i^*|Y,J_s,\Psi_{i-1}^*,\Psi^{i+1})$ is the proposal density, $\text{E}_1$ is the expectaion with respect to the distribution $\pi(\theta_i,\Psi^{i+1}|Y,J_s,\Psi_{i-1}^*)$, and $\text{E}_2$ is that with respect to $\pi(\Psi^{i+1}|Y,J_s,\Psi_{i-1}^*,\theta_i^*)\times q(\theta_i^*,\theta_i|Y,J_s,\Psi_{i-1}^*,\Psi^{i+1})$.
\par Then an estimate of $\pi(\theta_i|Y,J_s,\theta_1^*,\dots,\theta_{i-1}^*)$ is given by
\begin{align}
	\label{eqn:eqlabel}
	\begin{split}
		&\hat{\pi}(\theta_i|Y,J_s,\theta_1^*,\dots,\theta_{i-1}^*)\\
		&=\frac{G^{-1}\sum_{g=1}^{G}\rho(\theta_i^{(g)},\theta_i^*|Y,J_s,\Psi_{i-1}^*,\Psi^{i+1,(g)})q(\theta_i^{(g)},\theta_i^*|Y,J_s,\Psi_{i-1}^*,\Psi^{i+1,(g)})}{M^{-1}\sum_{m=1}^{M}\rho(\theta_i^*,\theta_i^{(m)}|Y,J_s,\Psi_{i-1}^*,\Psi^{i+1,(m)})},
	\end{split}
\end{align}
where $\{\theta_i^{(g)},\Psi^{i+1,(g)}\}_{g=1}^{G}$ are obtained from $\pi(\theta_i,\Psi^{i+1}|Y,J_s,\Psi_{i-1}^*)$. $\{\Psi^{i+1,(m)}\}$ are obtained from  $\pi(\Psi^{i+1}|Y,J_s,\Psi_{i-1}^*,\theta_i^*)$, and then for each $\{\Psi^{i+1,(m)}\}$, sample $\theta_i^{(m)}$ from $q(\theta_i^*, \theta_i|Y,J_s,\Psi^{i+1,(m)})$.
%------------------------------------------------

\section{Posterior Contraction Rate}
\par For classification problem, the most important object to study is the misclassification rate. By examining convergence to the true distribution, it follows that the Bayes procedure has misclassification rate close to that of the oracle procedure which uses the true values of the regression functions and other parameters (if any), e.g., cut-points in the ordered multinomial probit model. In the Bayesian nonparametric setting, Hellinger convergence is established by applying the general theory \citep{Ghosal}. Thus, in this section, we only consider the contraction rate of the posterior distribution with respect to a metric on the probability of categories, which is equivalent with Hellinger distance on the joint distribution. The posterior contraction rates of the three multinomial models with finite random series prior can be obtained using calculation similar to those in \cite{Shen} on posterior contraction rates for finite random series. 
\par \sloppy We use $\lesssim$ to denote an inequality up to a constant multiple, $f\asymp g$ for $f\lesssim g\lesssim f$.  For a vector $\theta \in \R^d, \| \theta \| _{p}=\{\sum_{i=1}^{d}|\theta_{i}|^p\}^{1/p}$, where $1\leq p<\infty$, and $\| \theta\|_{\infty}=\text{max}_{1\leq i \leq d}|\theta_{i}|$. Similarly, for a function $f$ with respect to a measure $G$, we define $\|f\|_{p,G}=\{\int |f(x)|^pdG\}^{1/p}$, where $1\leq p<\infty$, and $\|f\|_{\infty,G}= \text{sup}_x|f(x)|$. Let $\Na(\epsilon,T,d)$ denote the $\epsilon$-covering number of a set $T$ for a metric $d$. Let $h^2(p,q)=\int (\sqrt{p}-\sqrt{q})^2d\mu$ be the squared Hellinger distance, $K(p,q)=\int p\log(p/q)d\mu$, $V(p,q)=\int p\log^2(p/q)d\mu$ be the Kullback-Leibler (KL) divergences.
\par Suppose that $(X_i,Y_i)$, $i=1,\dots,n$, are the independent observations. Let $p$ denote the joint probability of $(X,Y)$, where $Y$ takes values $1,\dots,K$, and $p_0$ denote the true joint probabilty. Let $(X^{(n)},Y^{(n)})$ be the vector of $n$ obeservations following the probability $p^{(n)}$. Let $\pi_{k}(X)=\text{P}(Y=k|X)$ be the probability of the $k$th category conditioned on $X$, and $\pi_{0k}$ be the true probablity of the $k$th category conditioned on $X$. Define the probability vector $\pi=(\pi_{1},\dots,\pi_{K})^T$, where $\pi_{K}=1-\sum_{k=1}^{K-1}\pi_{k}$, and $\pi_{0}=(\pi_{01},\dots,\pi_{0K})^T$, where $\pi_{0K}=1-\sum_{k=1}^{K-1}\pi_{0k}$. Assume that the joint distribution of $(X,Y)$ follows $\nu\times G$, where $\nu$ denotes the counting measure on $\{1,\dots,K\}$. For these multinomial models, the KL divergences $K(p_0,p)$, and $V(p_0,p)$ can be reduced to 
\begin{align}
\label{eqn:eqlabel}
\begin{split}
K(p_0,p)&=\int\int p_0(x,y)\log\frac{p_0(x,y)}{p(x,y)}\dd{\nu(y)}\dd{x}\\
&=\int\int{\pi_0}(y|x)\log\frac{{\pi_0}(y|x)}{{\pi}(y|x)}\dd{\nu(y)}\dd{G(x)}\\
&=\text{E}_{X}\Bigl\{\sum_{k=1}^{K}\pi_{0k}(X)\log\frac{\pi_{0k}(X)}{\pi_{k}(X)}\Bigr\}\\
&=K(\pi_0,\pi), \: \text{say}
\end{split}
\end{align}
\begin{align}
\label{eqn:eqlabel}
\begin{split}
V(p_0,p)&=\int\int p_0(x,y)\log^2\frac{p_0(x,y)}{p(x,y)}\dd{\nu(y)}\dd{x}\\
&=\int\int{\pi_0}(y|x)\log^2\frac{{\pi_0}(y|x)}{{\pi}(y|x)}\dd{\nu(y)}\dd{G(x)}\\
&=\text{E}_{X}\Bigl\{\sum_{k=1}^{K}\pi_{0k}(X)\log^2\frac{\pi_{0k}(X)}{\pi_{k}(X)}\Bigr\}\\
&=V(\pi_0,\pi), \: \text{say}.
\end{split}
\end{align}
Similarly, the squared Hellinger distance $h^2(p_1,p_2)$ can be reduced to 
\begin{align}
\label{eqn:eqlabel}
\begin{split}
h^2(p_1,p_2)&=\int\int \big(\sqrt{{\pi_1}(y|x)}-\sqrt{{\pi_2}(y|x)} \big)\dd{\nu(y)}\dd{G(x)}\\
&=\text{E}_{X}\Bigl\{\sum_{k=1}^{K}\big(\sqrt{\pi_{1k}(X)}-\sqrt{\pi_{2k}(X)}\big)^2\Bigr\}\\
&=h^2(\pi_1,\pi_2), \: \text{say}.
\end{split}
\end{align}
\par We define a metric by
\begin{align}
	d(\pi,\pi_0)=\sqrt{\sum^{K}_{k=1}\text{E}_X|\pi_{k}(X)-\pi_{0k}(X)|^2}.
\end{align}
Then we have the following slightly simplification of a general posterior contraction theorem suitable in our context.
\begin{theorem}
	Assume that $\pi_0$ is bounded away from zero. Let $\epsilon_n\geq \bar{\epsilon}_n$ be two sequences of positive numbers satisfying $\epsilon_n\rightarrow 0$ and $n\bar{\epsilon}^2_n\rightarrow \infty$. Let $\Xa_0$ be such that $\text{P}(X\in\Xa_0)=1$ and $\pi_k(x),k=1,\dots,K$ for $x\in\Xa_0$ is bounded away from $0$. Let $\Wa_n$ be a subset of the parameter space such that the following conditions hold for some positive constants $a_2$ and $a_1>a_2+2$:
	\begin{align}
	&\log\Na(\epsilon_n, \Wa_n, h)\lesssim n\epsilon^2_n,\\
	&\Pi(\pi\not\in\Wa_n)\leq \exp\{-a_1n\bar{\epsilon}^2_n\},\\
	&-\log\Pi\left(\sum_{k=1}^{K}\|\pi_k-\pi_{0k}\|_{\infty,\Xa_0}^2\leq \bar{\epsilon}^2_n\right)\le a_2 n\bar{\epsilon}^2_n,
	\end{align}
	where $\|\pi_k-\pi_{0k}\|_{\infty,\Xa_0}=\sup_{x\in\Xa_0}\left|\pi_k(x)-\pi_{0k}(x)\right|$. Then for every $M_n\rightarrow \infty$, we have $\Pi\left(d(\pi,\pi_0)\geq M_n\epsilon_n|X^{(n)},Y^{(n)}\right)\rightarrow 0$ in probability.
\end{theorem}
\par The proof follows from Theorem 4 of \cite{Ghosal2007a}, by observing that 
	\begin{align}
	h^2(\pi,\pi_0)=\text{E}_{X}\sum_{k=1}^{K}\frac{|\pi_{k}(X)-\pi_{0k}(X)|^2}{|\sqrt{\pi_{k}(X)}+\sqrt{\pi_{0k}(X)}|^2}\gtrsim\text{E}_{X}\sum_{k=1}^{K}|\pi_{k}(X)-\pi_{0k}(X)|^2,
	\end{align}

and by expanding in Taylor's expansion
	\begin{align}
	\text{max}\{K(\pi_0,\pi),V(\pi_0,\pi)\}\lesssim\sum_{k=1}^{K}\|\pi_k-\pi_{0k}\|^2_{\infty,\Xa_0}.
	\end{align}

%, and the parameter blocks $B$, where $B=(\theta, \gamma)$ for ordered multinomial probit model, and $B=(\theta_1,\dots,\theta_{K-1})$ for unordered multinomial probit and multinomial logistic models

\par Let $\Pi$ be a generic notation for priors on the number $J$ of basis functions. As in \cite{Shen}, the priors on $J$ and the coefficients of the basis functions $\theta=(\theta_1,\dots,\theta_J)^T$ need to satisfy the conditons (A1) and (A2). For the ordered multinomial probit model, we add condition (A3).
\begin{enumerate}[label=(A\arabic*)]
	\item For some $c_1, c_2>0$, $0\leq t_2\leq t_1\leq 1$, $\exp\{-c_1j\log^{t_1}j\}\leq \Pi(J=j)\leq \exp\{-c_2j\log^{t_2}j\}$.
	\item Given $J$, $\Pi(\|\theta-\theta_0\|_{2}\leq \epsilon)\geq\exp\{-c_3J\log(1/\epsilon)\}$ for every $\|\theta_0\|_{\infty}\leq H$, where $c_3$ is some positive constant, $H$ is chosen sufficiently large, and $\epsilon>0$ is sufficiently small. Also, assume that $\Pi(\theta\not\in[-M,M]^J)\leq J\exp\{-CM^{t_3}\}$ for some cosntant $C$, $t_3>0$,
	\item Given $K$ categories, $\Pi(\|\gamma-\gamma_0\|_{2}\leq \epsilon)\geq\exp\{-c_4K\log(1/\epsilon)\}$, where $c_4$ is some positive constant.
\end{enumerate}
Geometric distribution with $t_1=t_2=0$, and Poisson distribution with $t_1=t_2=1$ on $J$ satisfy (A1). The multivariate normal distribution on $\theta$ and $\gamma$ satisfy (A2) and (A3) respectively. 
\par To obtain the posterior contraction rate, we need to verify the conditions (5.5)--(5.7), and we also need additional assumptions on the basis. We use $\theta^T\psi(t)$ to approximate $\beta(t)$, where $\theta=(\theta_1,\dots,\theta_J)^T$, and $\psi(t)=(\psi_1(t),\dots,\psi_J(t))^T$. Let $\beta_0(t)$ be the true value, and $r=2$ or $\infty$. Assume that there exist a $\theta_0\in \R^J$, $\|\theta_0\|_{\infty}\leq H$ and $K_0\geq 0$ such that
\begin{align}
	&\|\beta_0(\cdot)-\theta_0^T\psi(\cdot)\|_r\lesssim J^{-\alpha},\\
	&\|\theta_1^T\psi(\cdot)-\theta_2^T\psi(\cdot)\|_r\lesssim J^{K_0}\|\theta_1-\theta_2\|_2,\: \theta_1,\theta_2\in \R^J.
\end{align}
\par Remark 2 of \cite{Shen} gave examples of bases satisfying relations (5.10) and (5.11). For B-splines, the relations hold when $K_0=1/2$ with $r=2$, and $K_0=1$ with $r=\infty$.
\begin{remark}
\normalfont Parameter estimation plays a secondary role here. The problem of estimating model parameters is interesting in its own right but is not necessary for good classifications. \cite{Cai} and \cite{Yuan} showed that the parameter function estimation and the prediction from an estimator of the parameter function have different characteristics.
\end{remark}

\subsection{Ordered Multinomial Probit Model}
\par Let $\gamma=(\gamma_{1},\dots,\gamma_{K})^T$ be the vector of the threshold points, and $\gamma_{0}=(\gamma_{01},\dots,\gamma_{0K})^T$ be the vector of the true values of the threshold points. Let $\beta(t)$ be the parameter function on $[0,1]$, and $\beta_{0}(t)$ be the true parameter function on $[0,1]$. Let 
\begin{align}
\pi_k(X)=\Phi\Big(\gamma_{k}-\int \beta(t)X(t)dt\Big)-\Phi\Big(\gamma_{k-1}-\int \beta(t)X(t)dt\Big),
\end{align}
 and 
 \begin{align}
 \pi_{0k}(X)=\Phi\Big(\gamma_{0k}-\int \beta_{0}(t)X(t)dt\Big)-\Phi\Big(\gamma_{0k-1}-\int \beta_{0}(t)X(t)dt\Big).
\end{align} 

\begin{theorem}
 Assume that $\|X\|_1=\int |X(t)|\dd{t}$ is a bounded random variable, the priors satisfy the conditons (A1), (A2) and (A3), and that the basis $\psi(t)$ satisfies (5.10) and (5.11) with $r=\infty$. Then the posterior contraction rate of the ordered multinomial probit model is $\epsilon_{n}\asymp n^{-\alpha/(2\alpha+1)}(\log n)^{\alpha/(2\alpha+1)+(1-t_{2})/2}$ relative to $d(\pi,\pi_0)$. More explicitly, then for every $M_n\rightarrow \infty$, $\Pi(\beta: \rho(\beta,\beta_0)\geq M_n\epsilon_n|X^{(n)},Y^{(n)})\rightarrow0$ in probability, where $\rho(\beta,\beta_0)=\normalfont \text{E}_X|\int(\beta(t)-\beta_0(t))X(t)\dd{t}|$, and $\Pi(\gamma: \normalfont\max_j |\gamma_j-\gamma_{0j}|\geq M_n\epsilon_n|X^{(n)},Y^{(n)})\rightarrow0$ in probability.
\end{theorem}
\begin{proof}
	For any $x\in\Xa_0=\{\int |X(t)|\dd{t}\leq M\}$, say, by the Lipschitz continuity of $\Phi$, we have
	\begin{align}
		\label{eqn:eqlabel}
		\begin{split}
		    \left|\pi_k(x)-\pi_{0k}(x)\right|
			&\lesssim \max_k\left|\gamma_k-\gamma_{0k}\right|+\left|\int (\beta(t)-\beta_0(t)x(t)dt\right|\\
			&\lesssim\|\gamma-\gamma_0\|_{\infty}+\|\beta(\cdot)-\beta_{0}(\cdot)\|_{\infty}\int\left|x(t)\right|dt\\
			&\lesssim\|\gamma-\gamma_0\|_{\infty}+\|\beta(\cdot)-\beta_{0}(\cdot)\|_{\infty}.
		\end{split}
	\end{align}
	\par Observe that with the finite random series prior, the $L_{\infty}$-distance between $\beta(\cdot)$ and $\beta_{0}(\cdot)$ is bounded by
	\begin{align}
		\label{eqn:eqlabel}
		\begin{split}
			\|\beta(\cdot)-\beta_{0}(\cdot)\|_{\infty}&=\|\theta^T\psi(\cdot)-\theta_{0}^T\psi(\cdot)+\theta_{0}^T\psi(\cdot)-\beta_{0}(\cdot)\|_{\infty}\\
			&\leq\|\theta^T\psi(\cdot)-\theta_{0}^T\psi(\cdot)\|_{\infty}+\|\theta_{0}^T\psi(\cdot)-\beta_{0}(\cdot)\|_{\infty}.
		\end{split}
	\end{align}
	Then we have
	\begin{align}
		\label{eqn:eqlabel}
		\begin{split}
			\Pi&\Bigl(\sum_{k=1}^{K}\|\pi_{k}-\pi_{0k}\|^2_{\infty,\Xa_0}\leq\bar{\epsilon}_{n}^2\Bigr)\\
			&\geq \Pi(\|\gamma-\gamma_0\|\leq\bar{\epsilon}_{n}/\sqrt{2})\Pi(\|\beta(\cdot)-\beta_0(\cdot)\|_{\infty}\leq\bar{\epsilon}_{n}/\sqrt{2})\\
			&\geq \Pi(\|\gamma-\gamma_0\|\leq\bar{\epsilon}_{n}/\sqrt{2})\Pi(\|\theta-\theta_0\|\leq\bar{\epsilon}_{n}/(2\sqrt{2}\bar{J_n}^{K_0}))\\
			&\gtrsim \exp\bigl\{-K\log({\sqrt{2}}/{\bar{\epsilon}_n})\bigr\}\exp\bigl\{-\bar{J_n}\log({2\sqrt{2}\bar{J_n}^{K_0}}/{\bar{\epsilon}_n})\bigr\}.
		\end{split}
	\end{align}
	To satisfy the relation (5.7), we need $\bar{J}_{n}^{-\alpha}\lesssim \bar{\epsilon}_{n}$ and
	\begin{align}
		K\log(\sqrt{2}/\bar{\epsilon}_n)+\bar{J_n}\log(2\sqrt{2}\bar{J_n}^{K_0}/\bar{\epsilon}_n)\lesssim n\bar{\epsilon}_n^2.
	\end{align}
	Thus (5.17) leads to the conditions that $\bar{J}_{n}\log n\lesssim n\bar{\epsilon}_{n}^2$. Then we obtain the preliminary contraction rate $\bar{\epsilon}_{n}\asymp n^{-\alpha/(2\alpha+1)}(\log n)^{\alpha/(2\alpha+1)}$, for $\bar{J}_{n}\asymp (n/\log n)^{1/(2\alpha+1)}$.
	\par Using (5.14), we obtain
	\begin{align}
		\log\Na(\epsilon_n, \Wa_{n}, h)\lesssim \log\Na(\epsilon_n, \Wa_{n}, \|\cdot\|_{\infty})\lesssim n\epsilon^2_n.
	\end{align}
	According to Theorem 2 of \cite{Shen}, to satify (5.18), we need
	\begin{align}
		J_n\{(K_0+1)\log J_n+\log M_n+C_0\log n\}\leq n\epsilon^2_n,
	\end{align}
	for some positive constant $C_0$. To satify (5.6), we need
	\begin{align}
		bn\bar{\epsilon}^2_n\leq J_n\log^{t_2}J_n,\: \log J_n+n\bar{\epsilon}^2_n\leq M^{t_3}_n,
	\end{align}
	for some $b>0$. For $M_{n}=n^{1/t_{3}}$, (5.20) implies that $J_{n}\log^{t_{2}}n\gtrsim n\bar{\epsilon}_{n}^{2}$. Thus $J_{n}\asymp n^{1/(2\alpha+1)}(\log n)^{2\alpha/(2\alpha+1)-t_{2}}$. Relation (5.19) implies that $J_{n}\log n\lesssim n\epsilon_{n}^2$.
	As a result, the posterior contraction rate is $\epsilon_{n}\asymp n^{-\alpha/(2\alpha+1)}(\log n)^{\alpha/(2\alpha+1)+(1-t_{2})/2}$ relative to $d(\pi,\pi_0)$.
	\par Further, by Jensen's inequality, we have
	\begin{align}
	\text{E}_X|\pi_{k}(X)-\pi_{0k}(X)|^2\geq\Big\{\text{E}_X|\pi_{k}(X)-\pi_{0k}(X)|\Big\}^2.
	\end{align}
	If $k=1$, by the mean value theorem and the uniform positivity of $\Phi$ on compact interval, then 
	\begin{align}
		\label{eqn:eqlabel}
		\begin{split}
		\text{E}_X|\pi_{1}(X)-\pi_{01}(k)|&=\text{E}_X\left|\Phi(-\int\beta(t)X(t)dt)-\Phi(-\int\beta_0(t)X(t)dt)\right|\\
			&\gtrsim \text{E}_X\left|\int\beta(t)X(t)dt-\int\beta_0(t)X(t)dt\right|.
		\end{split}
	\end{align}
	Hence if $\text{E}_X|\pi_{1}(X)-\pi_{01}(X)|^2$ is small, then $\text{E}_X\left|\int\beta(t)X(t)dt-\int\beta_0(t)X(t)dt\right|$ is also small. If $k=2$, we have
	\begin{align}
	\label{eqn:eqlabel}
	\begin{split}
     	\text{E}_X|\pi_{2}(X)-\pi_{02}(k)|&=\text{E}_X\left|\Phi(\gamma_2-\int\beta(t)X(t)dt)-\Phi(\gamma_{02}-\int\beta_0(t)X(t)dt)\right.\\
     &\left.\phantom{=\text{E}_X} -\Phi(-\int\beta(t)X(t)dt)+\Phi(-\int\beta_0(t)X(t)dt)\right|\\
     &\gtrsim\text{E}_X\left|\Phi(\gamma_2-\int\beta(t)X(t)dt)-\Phi(\gamma_{02}-\int\beta_0(t)X(t)dt)\right|\\
     &\phantom{\gtrsim\text{E}_X} -\text{E}_X\left|\Phi(-\int\beta(t)X(t)dt)-\Phi(-\int\beta_0(t)X(t)dt)\right|.
    \end{split}
	\end{align}
	From (5.22), we know that $\text{E}_X\left|\Phi(-\int\beta(t)X(t)dt)-\Phi(-\int\beta_0(t)X(t)dt)\right|$ is small, and if $\text{E}_X|\pi_{2}(X)-\pi_{02}(X)|^2$ is small, then
	\begin{align}
	 \text{E}_X\left|\Phi(\gamma_2-\int\beta(t)X(t)dt)-\Phi(\gamma_{02}-\int\beta_0(t)X(t)dt)\right|
	\end{align}
	is also small. By the mean value theorem and the uniform positivity of $\Phi$ on compact interval, we have
	\begin{align}
	\label{eqn:eqlabel}
	\begin{split}
	&\text{E}_X\left|\Phi(\gamma_2-\int\beta(t)X(t)dt)-\Phi(\gamma_{02}-\int\beta_0(t)X(t)dt)\right|\\
	&\gtrsim\text{E}_X\left|\gamma_2-\gamma_{02}-\int\beta(t)X(t)dt+\int\beta_0(t)X(t)dt\right|\\
	&\gtrsim\left|\gamma_2-\gamma_{02}\right|-\text{E}_X\left|\int\beta(t)X(t)dt-\int\beta_0(t)X(t)dt\right|.
	\end{split}
	\end{align}
	Hence $|\gamma_2-\gamma_{02}|$ is small. Similarly, we can prove that $|\gamma_k-\gamma_{0k}|$ is small for any $k$.
\end{proof}

\subsection{Unordered Multinomial Probit Model}
\begin{comment}\par The probability of the $k$th category is expressed by the latent variable. Let $\beta_{l}(t)=\beta_{j}'(t)-\beta_{k}'(t)$, and $\varepsilon_{l}=\varepsilon_{j}'-\varepsilon_{k}'$. Then $W_{l}=\int \beta_{l}(t)X(t)dt+\varepsilon_{l}$, where $l=j$ if $j<k$, $l=j-1$ if $j>k$ so that $l=1,\dots,K-1$.
\begin{align}
	\label{eqn:eqlabel}
	\begin{split}
		\pi_k(X)&=\text{P}(W_{1}\leq0,\dots,W_{K-1}\leq0)\\
		&=\text{P}\Big(\varepsilon_{1}\leq-\int \beta_{1}(t)X(t)dt,\dots,\varepsilon_{K-1}\leq-\int \beta_{K-1}(t)X(t)dt\Big)\\
		&=\frac{1}{(2\pi)^{(K-1)/2}|\Sigma|^{1/2}}\\
		&\times\int_{-\infty}^{-\int \beta_{1}(t)X(t)dt}\dots\int_{-\infty}^{-\int \beta_{K-1}(t)X(t)dt}\exp\big(-\frac{1}{2}Z^T\Sigma^{-1}Z\big)dZ,
	\end{split}
\end{align}
where $\Sigma$ is a $(K-1)\times(K-1)$ matrix with 2 on diagonal entries and 1 on all off-diagonal entries.
\par Due to the exchangeable correlation structure of $\Sigma$, $\pi_k$ can be reduced to a one-dimensional integral \citep{Dunnett}
\end{comment}
Note that by (4.14)
\begin{align}
	\label{eqn:eqlabel}
	\begin{split}
		\pi_k(X)=\frac{1}{\sqrt{\pi}}\int_{0}^{\infty}&\Bigl\{\prod_{l=1}^{K-1}\Phi(-z\sqrt{2}-\int\beta_{l}(t)X(t)dt)\\
		&+\prod_{l=1}^{K-1}\Phi(z\sqrt{2}-\int\beta_{l}(t)X(t)dt)\Bigr\}e^{-z^2}dz
	\end{split}
\end{align}
\\
\begin{theorem}
	\sloppy Assume that $\|X\|_1=\int |X(t)|\dd{t}$ is a bounded random variable, the priors satisfy the conditons (A1) and (A2), and that the basis $\psi(t)$ satisfies (5.10) and (5.11) with $r=\infty$. Then the posterior contraction rate of the unordered multinomial probit model is $\epsilon_{n}\asymp n^{-\alpha/(2\alpha+1)}(\log n)^{\alpha/(2\alpha+1)+(1-t_{2})/2}$ relative to $d(\pi,\pi_0)$.
\end{theorem}
\begin{proof}
	For some $M>0$, $\text{P}(\Xa_0)=1$ for $\Xa_0=\{\int\left|X(t)\right|\dd{t}\leq M\}$. For any $x\in\Xa_0$, by the Lipschitz continuity of the function $\Phi$, we have
	\begin{align}
		\label{eqn:eqlabel}
		\begin{split}
			|\pi_{k}(x)-\pi_{0k}(x)|
			&\lesssim \left|\int \beta_k(t)x(t)dt-\int \beta_{0k}(t)x(t)dt\right|\\
			&\lesssim \int \left|\beta_k(t)-\beta_{0k}(t)\right| \left|x(t)\right|dt\\
			&\lesssim \|\beta_k(\cdot)-\beta_{0k}(\cdot)\|_{\infty}.\\
		\end{split}
	\end{align}
	\par The $L_{\infty}$-distance between $\beta_k(\cdot)$ and $\beta_{0k}(\cdot)$ is bounded by
	\begin{align}
		\label{eqn:eqlabel}
		\begin{split}
			\|\beta_k(\cdot)-\beta_{0k}(\cdot)\|_{\infty}&=\|\theta_k^T\psi(\cdot)-\theta_{0k}^T\psi(\cdot)+\theta_{0k}^T\psi(\cdot)-\beta_{0k}(\cdot)\|_{\infty}\\
			&\leq\|\theta_k^T\psi(\cdot)-\theta_{0k}^T\psi(\cdot)\|_{\infty}+\|\theta_{0k}^T\psi(\cdot)-\beta_{0k}(\cdot)\|_{\infty}.
		\end{split}
	\end{align}
	Then we have
	\begin{align}
		\label{eqn:eqlabel}
		\begin{split}
			\Pi\Big(\sum_{k=1}^{K}\|\pi_{k}-\pi_{0k}\|^2_{\infty,\Xa_0}\leq\bar{\epsilon}_{n}^2\Big)
			&\geq \Pi\Big(\sum_{k=1}^{K}\|\beta_k(\cdot)-\beta_{0k}(\cdot)\|^2_{\infty}\leq\bar{\epsilon}_{n}^2\Big)\\
			&\geq \Pi(\|\theta_k-\theta_{0k}\|\leq\bar{\epsilon}_{n}/(2\sqrt{K}\bar{J_n}^{K_0}))\\
			&\gtrsim \exp\bigl\{-\bar{J_n}\log({2\sqrt{K}\bar{J_n}^{K_0}}/{\bar{\epsilon}_n})\bigr\}.
		\end{split}
	\end{align}
	To satisfy the relation (5.7), we need  $\bar{J}_{n}^{-\alpha}\lesssim \bar{\epsilon}_{n}$ and
	\begin{align}
		\bar{J_n}\log(2\sqrt{K}\bar{J_n}^{K_0}/\bar{\epsilon}_n)\lesssim n\bar{\epsilon}_n^2.
	\end{align}
	Thus, (5.30) leads to the conditions that $\bar{J}_{n}\log n\lesssim n\bar{\epsilon}_{n}^2$. Then we obtain the preliminary contraction rate $\bar{\epsilon}_{n}\asymp n^{-\alpha/(2\alpha+1)}(\log n)^{\alpha/(2\alpha+1)}$, for $\bar{J}_{n}\asymp (n/\log n)^{1/(2\alpha+1)}$.
	\par Following the same arguments as (5.18)--(5.20), the posterior contraction rate is $\epsilon_{n}\asymp n^{-\alpha/(2\alpha+1)}(\log n)^{\alpha/(2\alpha+1)+(1-t_{2})/2}$ relative to $d(\pi,\pi_0)$.
\end{proof}

\subsection{Multinomial Logistic Model}
\par Let $\beta_{k}(t)$, $k=1,\dots,K-1$, be the coefficient functions on $[0,1]$, and $\beta_{0k}(t)$, $k=1,\dots,K-1$, be the true coefficient functions on $[0,1]$.\\ 
\begin{theorem}
	\sloppy Assume that $\|X\|_1=\int |X(t)|\dd{t}$ is a bounded random variable, the priors satisfy the conditons (A1) and (A2), and that the basis $\psi(t)$ satisfies (5.10) and (5.11) with $r=\infty$. Then the posterior contraction rate of the multinomial logistic model is $\epsilon_{n}\asymp n^{-\alpha/(2\alpha+1)}(\log n)^{\alpha/(2\alpha+1)+(1-t_{2})/2}$ relative to $d(\pi,\pi_0)$.
\end{theorem}
\begin{proof}
	The proof is similar to that of Theorem 3.
\end{proof}

\section{Discriminant Analysis}
\par As a comparison to those multinomial models, we use Bayesian discriminant analysis to classify the functional data. Instead of modeling the class probability directly, the discriminant analysis uses Bayes's rule to compute the marginal likelihood of $Y_i$ \citep{Gelman}.
The classical discriminant analysis applies only to multivariate data. For functional data, we can use certain orthogonal linear functions to determine the classification probabilities:
\begin{align}
	(f_{i1},\dots,f_{im})^{T}=\left(\int \beta_{1}(t)X_{i}(t)dt,\dots,\int \beta_{m}(t)X_{i}(t)dt\right)^{T}
\end{align}
\par Ideally these $\beta_{1}(t),\dots,\beta_{m}(t)$ are unknown, but putting a prior on these functions with identifiability restrictions is complicated. We instead consider $\beta_{1}(t),\dots,\beta_{m}(t)$ to be known as the first $m$ principal components \citep{Ramsay}, but let the means and the covariance matrices be unknown. Then discriminant analysis can be applied to the $m$ principal components.

\subsection{Linear Discriminant Analysis} 
\par Linear discriminant analysis assumes that for each of the $K$ category, the set of linear function $(f_1,\dots,f_m)$ follows a normal distribution with the same covarince matrix: $(f_{il1},\dots,f_{ilm})^{T}\sim \text{N}(\mu_{l},\Sigma)$, where $\mu_{l}$ is the population mean of category $l$, $l=1,\dots,K$, $i=1,\dots,n_{l}$, and $n_{l}$ is the number of data in category $l$. Then the probability of choosing category $k$ is given by
\begin{align}
	\text{P}(Y_{i}=k|X_{i})=\frac{p_k\cdot\phi(f_{ik1},\dots,f_{ikm};\mu_{k},\Sigma)}{\sum_{l=1}^{K}p_l\cdot\phi(f_{il1},\dots,f_{ilm};\mu_{l},\Sigma)},
\end{align}
where $\phi(f_1,\dots,f_m;\mu,\Sigma)$ is the multivariate normal density function with mean $\mu$ and covariace $\Sigma$, and $p_l$, $l=1,\dots,K$, are the probability of choosing category $l$.
\par The variables $f_{il1},\dots,f_{ilm}$ are the $m$ principal components of $X_{i}(t)$ in categoty $l$, where $l=1,\dots,K$. Define $f_{il}=(f_{il1},\dots,f_{ilm})^{T}$, where $i=1,\dots,n_{l}$, and $\sum_{l=1}^{K}n_{l}=n$.
To estimate the mean $\mu_{l}$ for each category $l$, and the common covariance $\Sigma$ among all categories, we use the conjugate normal-inverse-Wishart prior with hyperparameters \citep{Gelman} for $(\mu_{l},\Sigma)$ 
\begin{align}
	\Sigma \sim \text{IW}_{\nu_{0}}(\Lambda_{0}^{-1}), \; \mu_{l}|\Sigma \sim \text{N}(\mu_{l0},\Sigma/\kappa_{0}).
\end{align}
Then the posterior distribution of $(\mu_{l},\Sigma)$ can be obtained in the following order
\begin{align}
	\Sigma|Y \sim \text{IW}_{\nu_{n}}(\Lambda_{n}^{-1}), \; \mu_{l}|\Sigma, Y \sim \text{N}(\mu_{ln},\Sigma/\kappa_{n}),
\end{align}
where $\nu_{n}=\nu_{0}+n$, $\bar{f_{l}}=\sum_{i=1}^{n_{l}}f_{il}/n_{l}$, $S=\sum_{l=1}^{K}\sum_{i=1}^{n_{l}}(f_{il}-\bar{f_{l}})(f_{il}-\bar{f_{l}})^{T}$,
\begin{align}
	\Lambda_{n}=\Lambda_{0}+S+\sum_{l=1}^{K}\frac{\kappa_{0}n_{l}}{\kappa_{0}+n_{l}}(\bar{f_{l}}-\mu_{l0})(\bar{f_{l}}-\mu_{l0})^{T},
\end{align}
and 
\begin{align}
	\kappa_{n}=\kappa_{0}+n,\: \mu_{ln}=\frac{\kappa_{0}\mu_{l0}+n_{l}\bar{f_{l}}}{\kappa_{0}+n_{l}}, \: l=1,\dots,K.
\end{align}

\subsection{Quadratic Discriminant Analysis}
\par Quadratic discriminant analysis is defined in a similar way, except that it has a different covariance matrix for each category. The probability of choosing category $k$ is given by
\begin{align}
	\text{P}(Y_{i}=k|X_{i})=\frac{p_k\cdot\phi(f_{ik1},\dots,f_{ikm};\mu_{k},\Sigma_{k})}{\sum_{l=1}^{K}p_l\cdot\phi(f_{il1},\dots,f_{ilm};\mu_{l},\Sigma_{l})}.
\end{align}
\par To estimate the mean $\mu_{l}$ and the covariance $\Sigma_{l}$ for each category $l$, where $l=1,\dots,K$, we use the conjugate normal-inverse-Wishart prior with hyperparameters for $(\mu_{l},\Sigma_{l})$
\begin{align}
	\Sigma_{l} \sim \text{IW}_{\nu_{l0}}(\Lambda_{l0}^{-1}), \quad \mu_{l}|\Sigma_{l} \sim \text{N}(\mu_{l0},\Sigma_{l}/\kappa_{l0}),
\end{align}
for $l=1,\dots,K$.
Then the posterior distribution of $(\mu_{l},\Sigma_{l})$ can be obtained in the following order
\begin{align}
	\Sigma_{l}|Y \sim \text{IW}_{\nu_{ln}}(\Lambda_{ln}^{-1}), \quad \mu_{l}|\Sigma_{l}, Y \sim \text{N}(\mu_{ln},\Sigma_{l}/\kappa_{ln}),
\end{align}
where $\nu_{ln}=\nu_{l0}+n_{l}$, $\bar{f_{l}}=\sum_{i=1}^{n_{l}}f_{il}/n_{l}$, $S_{l}=\sum_{i=1}^{n_{l}}(f_{il}-\bar{f_{l}})(f_{il}-\bar{f_{l}})^{T}$,
\begin{align}
	\Lambda_{ln}=\Lambda_{l0}+S_{l}+\frac{\kappa_{l0}n_{l}}{\kappa_{l0}+n_{l}}(\bar{f_{l}}-\mu_{l0})(\bar{f_{l}}-\mu_{l0})^{T},
\end{align}
and 
\begin{align}
	\kappa_{ln}=\kappa_{l0}+n_{l},\: \mu_{ln}=\frac{\kappa_{l0}\mu_{l0}+n_{l}\bar{f_{l}}}{\kappa_{l0}+n_{l}},\: l=1,\dots,K.
\end{align}

\section{Simulation}

\subsection{Data Generation}
%A statement requiring citation \cite{Figueredo:2009dg}.
\par The simulated data are generated following different data generating process. All of the simulated data have three categories. In all cases considered below, we generate the functional data from a Gaussian process at discrete time points 0, 0.01, $\dots$, 0.99, 1, with the mean function $\sin t$ and variation kernel $100\exp\{-100(t_i-t_j)^2\}$, where $t_i$ and $t_j$ were the discrete time point 0, 0.01, $\dots$, 0.99, 1.
\begin{figure*}
	\begin{subfigure}[b]{0.3\textwidth}
		\centering
		\includegraphics[width=\textwidth]{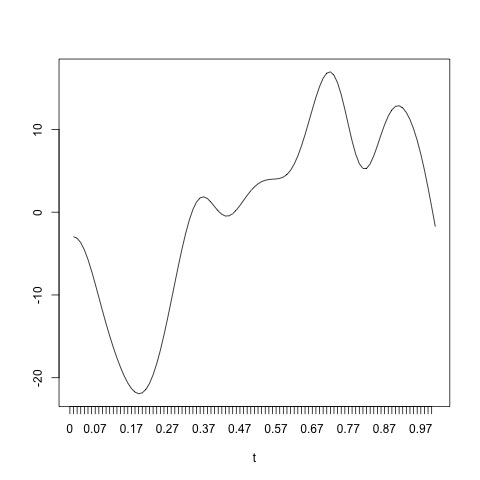}
		\caption{$\beta(t)$ for the ordered multinomial probit model}
	\end{subfigure}
	\quad
	\begin{subfigure}[b]{0.3\textwidth}
		\centering
		\includegraphics[width=\textwidth]{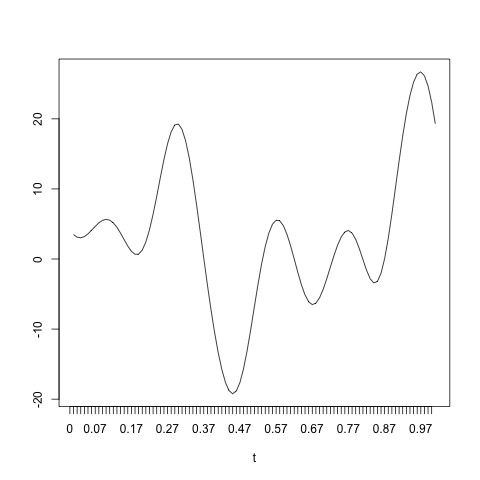}
		\caption{$\beta_1(t)$ for the unordered multinomial probit model}
	\end{subfigure}
	\quad
	\begin{subfigure}[b]{0.3\textwidth}
		\centering
		\includegraphics[width=\textwidth]{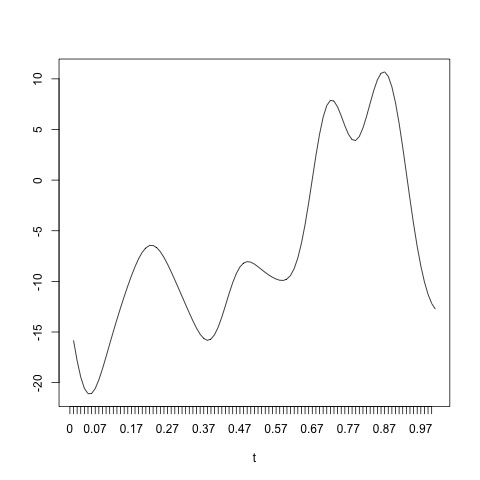}
		\caption{$\beta_2(t)$ for the unordered multinomial probit model}
	\end{subfigure}
	\vskip\baselineskip
	\begin{subfigure}[b]{0.3\textwidth}
		\centering
		\includegraphics[width=\textwidth]{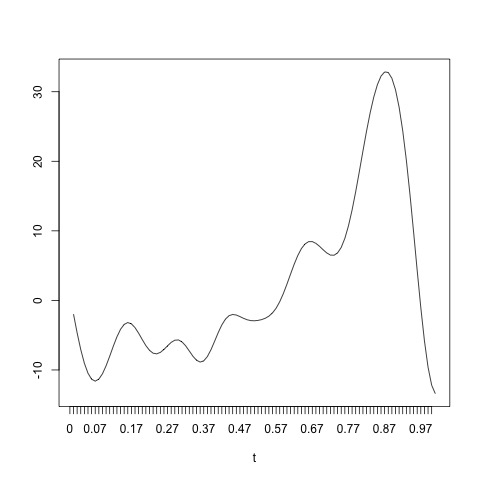}
		\caption{$\beta_3(t)$ for the unordered multinomial probit model}
	\end{subfigure}
	\quad
	\begin{subfigure}[b]{0.3\textwidth}
		\centering
		\includegraphics[width=\textwidth]{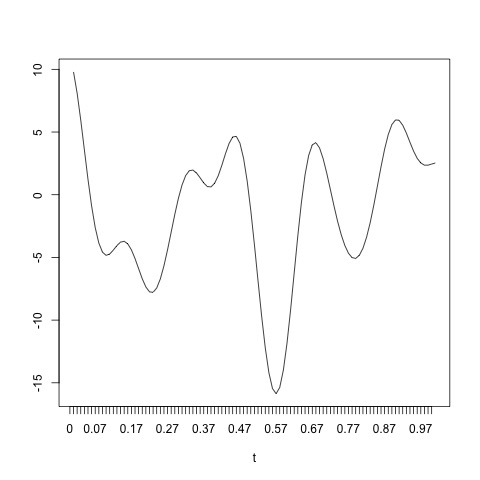}
		\caption{$\beta_1(t)$ for the multinomial logistic model}
	\end{subfigure}
	\quad
	\begin{subfigure}[b]{0.3\textwidth}
		\centering
		\includegraphics[width=\textwidth]{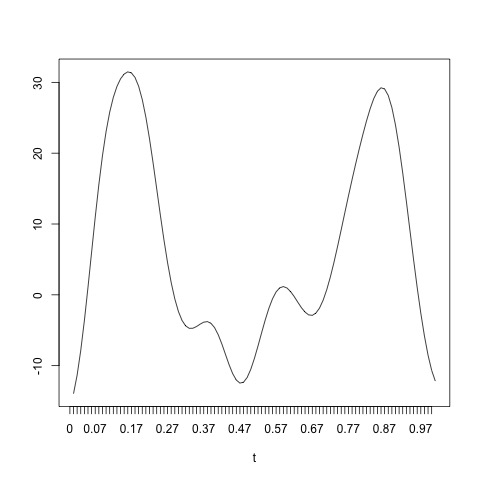}
		\caption{$\beta_2(t)$ for the multinomial logistic model}
	\end{subfigure}
	\quad
	\caption{Coefficient functions for the multinomial models}
\end{figure*}
\par For the ordered multinomial probit data, the coefficient function $\beta(t)$ is plotted in Figure 1 (a), and the four threshold points are chosen to be $-\infty$, 0, 8, $\infty$. The four cut-off points construct three intervals. If the inner product of a functional data and the coefficient function plus a standard normal variable falls in the $kth$ interval $(\gamma_{k-1},\gamma_k)$, then the functional data attributes to the category $k$. 
\par For unordered multinomial probit data, the coefficient functions $\beta_1(t), \beta_2(t), \beta_3(t)$ are plotted in Figure 1 (b)-(d). The inner product of a functional data and the three coefficient functions are added with standard normal variables, respectively. We sample from these three normal variables, and obtaine the corresponding probabilities. Then the functional data belonges to the category with the largest sampled value.
\par For the multinomial logistic data, the coefficient functions $\beta_1(t), \beta_2(t)$ are plotted in Figure 1 (e)-(f), and the third coefficient function $\beta_3(t)$ can be assumed to be zero everywhere. We compute the probability of a functional data falling into each category. Then the data attributes to the category with the largest probability.
\par For data satisfying the assumption of the linear discriminant analysis, we generate them from three Gaussian processes with different mean functions $\sin t+2\cos t$, $\sin t$, and $\sin t-3\cos t$, but the same variation kernel $\exp\{-30(t_i - t_j)^2\}$.
\par For data satisfying the assumption of the quadratic discriminant analysis, we generate them from three Gaussian processes with different mean functions and different variation kernels. The mean functions are $\sin t+2\cos t$, $\sin t$, and $\sin t-3\cos t$, and the variation kernels are $\exp\{-2\sin^2(\pi(t_i-t_j))\}$, $\exp\{-30(t_i - t_j)^2\}$, and $\exp\{-|t_i - t_j|\}$, respectively.
\par In this simulation study, we generate total 900 (300 for each category) functional data for each type of dataset. We constructe the training data with 720 (240 for each category) of them and the testing data with the remaining 180 (60 for each category) of them. 

\subsection{Basis Functions}
\par For models using the finite random series prior, we consider the B-spline basis. The B-spline basis functions on interval $[0, 1]$ can be created using the R package $\tt{fda}$. In this simulation study, we put a geometric prior with $p=0.5$ on $J$. We only consider the possible number of B-spline basis functions to be $J=5, \dots, 15$, since the probability outside this range is too small. Those B-spline basis functions are generated at the same discrete time points as the functional data, that is $0, 0.01, \dots, 0.99, 1$. 

\subsection{Results}
\par Under the chosen models, we apply Baysian estimation methods described in Section 3 on the training data. In this study, 5000 MCMC iterations are obtained, and the first 1000 of them are discarded as burn-in. We use the last 4000 MCMC output of the parameter $B$ to classify the 180 transformed testing data, where $B=(\theta,\gamma_{2},\gamma_{3})$ for the ordered multinomial probit model, $B=\Theta$ for the unorederd multinomial probit model, $B=(\theta_{1},\theta_{2})$ for the logistic model, $B=(\mu_{1},\mu_{2},\mu_{3},\Sigma)$ for the linear discriminant analysis model, and $B=(\mu_{1},\mu_{1},\mu_{3},\Sigma_{1},\Sigma_{2},\Sigma_{3})$ for the quadratic discriminant analysis model. A transformed testing data $z_{i}$ or $f_{i}$ is in categoty $k$ if $\sum_{g=1}^{4000}\mathbbm{1}(Y_{i}=k|z_{i} \ \text{or} \ f_{i}, B^{(g)})>\sum_{g=1}^{4000}\mathbbm{1}(Y_{i}=l|z_{i} \ \text{or} \ f_{i}, B^{(g)})$, where $l\neq k$. Then we use the techniques described in Section 4 to average the results from the multinomial models. As a comparison with the Bayesian method, the linear support vector machine (SVM) is also applied to the principal components of these training data, and made predictions on the testing data. To apply SVM, we use the R package $\tt{e1071}$. Table 1 shows the averaged misclassification rates for each data type under different models.

\renewcommand*{\arraystretch}{1.2}
\begin{longtable}{c@{\hspace{1.5\tabcolsep}}c@{\hspace{1.5\tabcolsep}}c@{\hspace{1.5\tabcolsep}}c@{\hspace{1.5\tabcolsep}}c@{\hspace{1.5\tabcolsep}}c@{\hspace{1.5\tabcolsep}}c}
	\caption{Averaged misclassification rates for simulated data}
	\endfirsthead
	\endhead
	\hline
	Dataset&OMP Model&UMP Model&MLO Model&LDA&QDA&SVM\\
	\hline
	OMP & $7.69\%$ & $30.56\%$  & $28.33\%$  & $38.89\%$ & $48.89\%$  & $15.00\%$ \\
	UMP & $38.96\%$ & $7.22\%$  & $7.78\%$  & $21.11\%$ & $21.11\%$  & $10.56\%$ \\
	MLO & $49.44\%$ & $4.75\%$  & $3.89\%$  & $32.22\%$ & $36.11\%$  & $7.78\%$ \\
	LDA & $26.32\%$ & $25.69\%$  & $26.11\%$  & $5.00\%$ & $5.00\%$  & $7.78\%$ \\
	QDA & $24.28\%$ & $21.95\%$  & $21.67\%$  & $10.56\%$ & $9.44\%$  & $8.33\%$ \\
	\hline
\end{longtable}

\section{Application}
\par We also test our models on phoneme data. This dataset can be found in the R package $\tt{fds}$, and can also be found at \url{https://www.math.univ-toulouse.fr/staph/npfda/}. The original data has 2000 $(X,Y)$ pairs, and five categories. For computational efficiency, we only use 900 of them from three categories. We split the data into training and testing set by randomly sampling from each class, and keeping the same percentage of samples of each class as the complete set. The size of the testing data is $20\%$ of the total data size. That is we have 240 data for each class in the training set, and 60 data for each class in the testing set. We put a geometric prior with $p=0.5$ on $J$, and it is enough for us to consider the number of B-spline basis functions to be $J=5, \dots, 15$. We obtain 5000 MCMC iterations and discard the first 1000 of them as burn-in.
\renewcommand*{\arraystretch}{1.2}
\begin{longtable}{ccccc}
	\caption{Averaged misclassification rates for phoneme data}
	\endfirsthead
	\endhead
	\hline
	OMP Model & UMP Model & MLO Model & LDA & QDA\\
	\hline
	$9.84\%$ & $0.56\%$  & $5.56\%$  & $7.78\%$ & $5.00\%$ \\
	\hline
\end{longtable}
\renewcommand*{\arraystretch}{1.2}
\begin{longtable}{ccc}
	\caption{Estimate and standard error of the posterior mean for the ordered multinomial model $(J=6)$}
	\endfirsthead
	\endhead
	\hline
	&$\gamma_2$&$\theta$\\
	\hline
	Estimate& $3.87$&$(52.12,-9.60,-8.89,-0.19,-4.91,2.85)$\\
	\hline
	Standard error & $0.03$ &$(0.34,0.11,0.13,0.08,0.10,0.10)$\\
	\hline
\end{longtable}
\par According to Table 2, the unordered multinomial probit model is the best model for the phoneme data. For this data, the categories are not naturally ordered, and hence ordered multinomial probit model is not natural for this problem, but we include it in the analysis for comparison. Figure 2 displays the cut-point $\gamma_2$ sampled by Metropolis-Hastings under different $J$, and we can tell that $\gamma_2$ converges around 500 iterations. Tables 3, 4, and 5 show the estimate and standard error of the posterior mean of the phoneme data under ordered multinomial probit model, multinomial logistic model, and unordered multinomial probit model, when $J=6$, $J=10$, and $J=14$, respectively. We choose these $J$ values because under these values the model has the largest posterior probability $\text{P}(J|Y)$. Although ordered multinomial probit model is not intuitive in this context, its performance is not too inferior.

\renewcommand*{\arraystretch}{1.2}
\begin{longtable}{cc}
	\caption{Estimate and standard error of the posterior mean for the multinomial logistic model $(J=10)$}
	\endfirsthead
	\endhead
	\hline
	&$\theta_2$\\
	\hline
	Estimate& $(13.10,18.25,6.04,-15.29,15.52,1.30,-5.81,4.65,-28.24,-16.91)$ \\
	\hline
	Standard error& $(0.94,1.08,0.64,0.63,0.75,0.66,0.37,0.48,0.70,1.03)$\\
	\hline
	&$\theta_3$\\
	\hline
	Estimate&$(39.42,34.30,-3.47,5.26,-7.36,0.38,-17.99,-4.43,-11.23,2.44)$\\
	\hline
	Standard error& $(1.21,1.44,0.42,0.31,1.08,0.27,0.53,0.33,0.93,0.32)$ \\
	\hline
\end{longtable}

\renewcommand*{\arraystretch}{1.2}
\begin{longtable}{ccc}
	\caption{Estimate and standard error of the posterior mean for the unordered multinomial model $(J=14)$}
	\endfirsthead
	\endhead
	\hline
	&estimate & standard error\\
	\hline
	$\Theta$& $\left[ \begin{array}{cc} -15.40 & 49.92  \\ 35.78 & 79.79 \\ 45.94 & 32.11\\ 4.97 & -0.76 \\-23.23 & -12.58 \\-15.09&-14.88\\23.43&-21.71\\-11.87&1.67\\-0.96&-5.06\\-0.27&-9.82\\1.58&-7.46\\-12.43&-14.68\\-28.97&-7.74\\-28.49&-3.38\end{array}\right]$ &$\left[ \begin{array}{cc} 0.93 & 0.85  \\ 0.52 & 0.81 \\ 0.58 & 0.75\\ 0.66 & 0.73 \\0.60 & 0.71 \\0.62&0.67\\0.68&0.73\\0.74&0.80\\0.63&0.64\\0.63&0.69\\0.70&0.78\\0.57&0.60\\0.64&0.69\\0.53&0.57\end{array}\right]$   \\
	\hline
\end{longtable}

\begin{figure*}
	\begin{subfigure}[b]{0.3\textwidth}
		\centering
		\includegraphics[width=\textwidth]{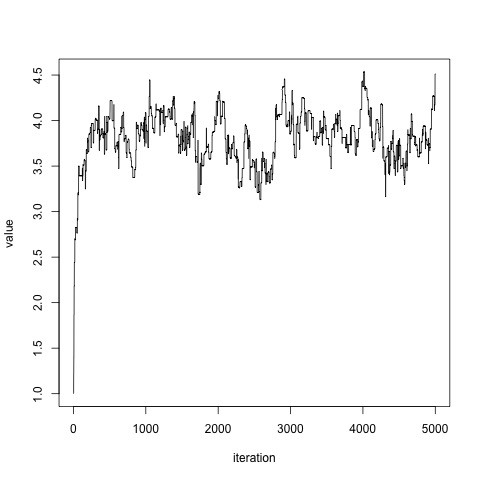}
		\caption{$J=5$}
	\end{subfigure}
	\quad
	\begin{subfigure}[b]{0.3\textwidth}
		\centering
		\includegraphics[width=\textwidth]{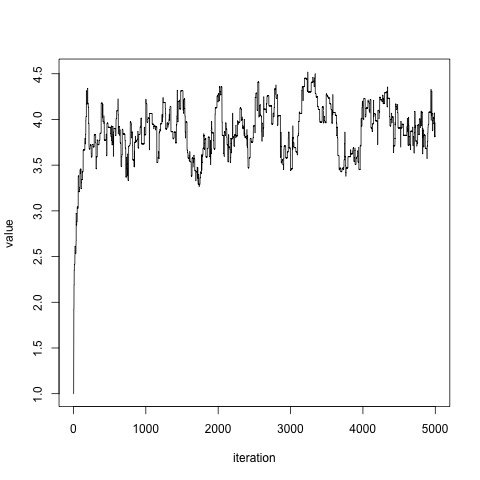}
		\caption{$J=7$}
	\end{subfigure}
	\quad
	\begin{subfigure}[b]{0.3\textwidth}
		\centering
		\includegraphics[width=\textwidth]{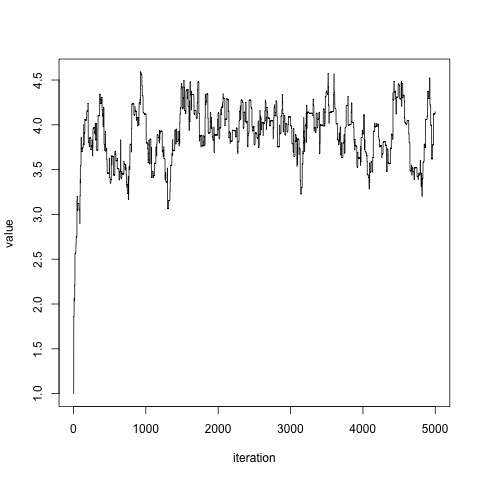}
		\caption{$J=9$}
	\end{subfigure}
	\vskip\baselineskip
	\begin{subfigure}[b]{0.3\textwidth}
		\centering
		\includegraphics[width=\textwidth]{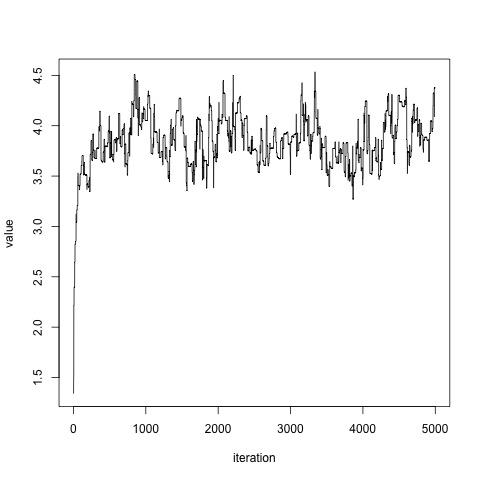}
		\caption{$J=11$}
	\end{subfigure}
	\quad
	\begin{subfigure}[b]{0.3\textwidth}
		\centering
		\includegraphics[width=\textwidth]{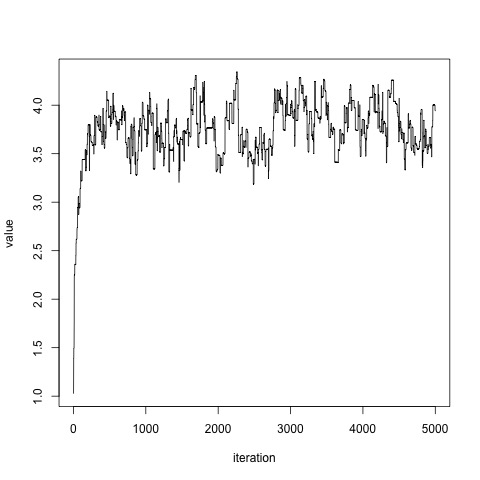}
		\caption{$J=13$}
	\end{subfigure}
	\quad
	\begin{subfigure}[b]{0.3\textwidth}
		\centering
		\includegraphics[width=\textwidth]{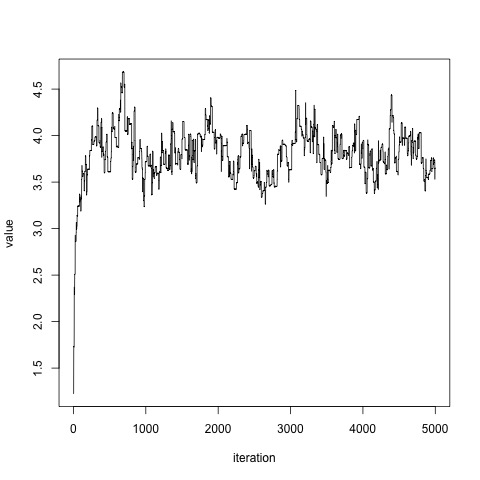}
		\caption{$J=15$}
	\end{subfigure}
	\quad
	\caption{$\gamma_2$ sampled from Metropolis-Hastings when $J$=5-7 and 13-15}
\end{figure*}
%	REFERENCE LIST
%----------------------------------------------------------------------------------------
\newpage
%\printbibliography
\bibliographystyle{imsart-nameyear}
\bibliography{ref}

\end{document}